\documentclass[11pt]{article}

\usepackage{amsfonts}
\usepackage{amsmath}
\usepackage{amssymb}
\usepackage{amsthm}

\usepackage{graphicx}

\theoremstyle{plain}
\newtheorem{theorem}{Theorem}

\newtheorem{lemma}[theorem]{Lemma}
\newtheorem{corollary}[theorem]{Corollary}

\theoremstyle{definition}

\newtheorem{remark}[theorem]{Remark}

\numberwithin{equation}{section}

\newcommand{\V}[1]{{\mathbf{#1}}}
\newcommand{\pdt}{\partial_t}

\begin{document}

\title{Epidemiological models with parametric heterogeneity:\\ Deterministic theory for closed populations}

\author{Artem S. Novozhilov\footnote{anovozhilov@gmail.com}\\[3mm]
\textit{\normalsize $^{1}$Applied Mathematics--1, Moscow State University of Railway Engineering,}\\[-1mm]\textit{\normalsize Obraztsova 9, bldg. 9, Moscow 127994, Russia}}
\date{}

\maketitle

\begin{abstract}
We present a unified mathematical approach to epidemiological models with parametric heterogeneity, i.e., to the models that describe individuals in the population as having specific parameter (trait) values that vary from one individuals to another. This is a natural framework to model, e.g., heterogeneity in susceptibility or infectivity of individuals. We review, along with the necessary theory, the results obtained using the discussed approach. In particular, we formulate and analyze an SIR model with distributed susceptibility and infectivity, showing that the epidemiological models for closed populations are well suited to the suggested framework. A number of known results from the literature is derived, including the final epidemic size equation for an SIR model with distributed susceptibility. It is proved that the bottom up approach of the theory of heterogeneous populations with parametric heterogeneity allows to infer the population level description, which was previously used without a firm mechanistic basis; in particular, the power law transmission function is shown to be a consequence of the initial gamma distributed susceptibility and infectivity. We discuss how the general theory can be applied to the modeling goals to include the heterogeneous contact population structure and provide analysis of an SI model with heterogeneous contacts. We conclude with a number of open questions and promising directions, where the theory of heterogeneous populations can lead to important simplifications and generalizations.

\end{abstract}

\section{Introduction}\label{section:1}
The real-world populations are \textit{heterogeneous}. The populations consist of individuals, and all the individuals are different. This is a basic fact, which does not require any proof. Individuals can differ in their age, spatial location, social habits, genome compositions, etc. Inasmuch as we aim to model the dynamics of interacting populations, it is an obvious condition to include this heterogeneous structure into the mathematical models. There are different kinds of population heterogeneity; here we only inspect heterogeneity in population parameters (such as, e.g., susceptibility to a specific disease); the parameters are considered as an inherent and invariant property of individuals, whereas the parameter values can vary between individuals. We call such heterogeneity \textit{parametric heterogeneity} (\cite{Dushoff1999}). Taking into account the parametric heterogeneity yields important changes for the population dynamics. Most importantly, it means that both \textit{evolutionary} and \textit{ecological} aspects of the dynamics have to be accounted for. In other words, we must follow not only the total population numbers (ecological aspects) but also the changes of the population structure (evolutionary aspects), which are described by the transformation of the parameter distribution in the population with time.

The simplest way to include the parametric heterogeneity into mathematical models is to divide populations into subgroups, such that each subgroup has its own specific parameter value. By increasing the number of groups we can assume that there is a continuous distribution of the parameter in the population.

To illustrate the mathematical approach to the parametric heterogeneity and present the notations used throughout the text, we start with the simplest possible mathematical model: the Malthus equation
\begin{equation}\label{eq1:1}
    \dot{N}=mN,
\end{equation}
where $N=N(t)$ is the total population size at time $t$. Equation \eqref{eq1:1} has the solution $N(t)=N(0)e^{mt}$ for the given initial condition $N(0)$.

Equation \eqref{eq1:1} has parameter $m$, which is the per capita growth rate (and is often called the Malthusian parameter). Model \eqref{eq1:1} can be used to describe, e.g., the growth of a bacterial colony. Employing equation \eqref{eq1:1} as a mathematical model to describe the changes of the population size notwithstanding, it is customary assumed that this parameter $m$ is the same for any individual in the population. This can be hardly true for any realistic population. The simplest way to describe the population which has different values of $m$ is to suppose that the total population consists of $k$ subpopulations, each of which has its own parameter value; i.e., we replace \eqref{eq1:1} with the following system of equations:
\begin{equation}\label{eq1:2}
    \dot{N}_i=m_iN_i,\quad i=1,\ldots,k,
\end{equation}
where now the total population size is $N(t)=\sum_i N_i(t)$. There are $k$ initial conditions for \eqref{eq1:2}: $N_i(0),\,i=1,\ldots,k$. Is it possible to describe the behavior of the total population size, whose evolution is governed by \eqref{eq1:2}, using only one equation? Summing all the equations in \eqref{eq1:2}, we obtain
\begin{equation}\label{eq1:3}
    \dot{N}=\sum_i m_iN_i=N\sum_i m_i\frac{N_i}{N}=\mathsf{E}_t[m]N=\bar{m}(t)N,
\end{equation}
where a natural notation used for the mean parameter value in the population at the time moment $t$:
$$
\bar{m}(t)=\mathsf{E}_t[m]=\sum_i m_ip_i(t),\quad p_i(t)=\frac{N_i(t)}{N(t)}\,,\quad \sum_i p_i(t)=1.
$$
Equation \eqref{eq1:3} is very similar to \eqref{eq1:1}, the difference is that it depends on the current parameter distribution in the population. If it is possible to find $\mathsf{E}_t[m]$, then the general solution to \eqref{eq1:3} is simply
\begin{equation}\label{eq1:4}
N(t)=N(0)\exp\{\int_0^t \mathsf{E}_\tau[m]\,d\tau\}
\end{equation}
for the initial population size $N(0)$.

An important disadvantage of the subgroup approach is that the heterogeneity within each group is not taken into account if a fixed number of groups is considered. As it was discussed, in case $k\to\infty$, we can conjecture that there is a limiting continuous parameter distribution. For the Malthus equation \eqref{eq1:1} this means that we suppose that each individual is characterized by its own parameter (trait) value $\omega$, such that the density of individuals with the given parameter value is $n(t,\omega)$, and the total population size is $N(t)=\int_\Omega n(t,\omega)\,d\omega$, where $\Omega$ is the set of the trait values. Now \eqref{eq1:3} is replaced with
\begin{equation}\label{eq1:5}
    \partial_t n(t,\omega)=m(\omega)n(t,\omega),
\end{equation}
where $\partial_t\equiv\frac{\partial}{\partial t}$. For \eqref{eq1:5} we need the initial condition $n(0,\omega)$. Using the notation $$p(t,\omega)=\frac{n(t,\omega)}{N(t)}$$ for the probability density function (pdf) of the current parameter distribution and the usual notation for the mean parameter value $$
\bar{m}(t)=\mathsf{E}_t[m]=\int_\Omega m(\omega)p(t,\omega)\,d\omega
$$
we find that the solution to \eqref{eq1:5} is again given by \eqref{eq1:4}. If one can calculate $\mathsf{E}_t[m]$ then the solution to \eqref{eq1:5} is found.

An important remark is worth spelling out. All the models we consider are deterministic. We use the probability theory language to describe evolution of the parameter distributions, however no stochastic effects are included into our models.

Models of the form \eqref{eq1:5} are infinite dimensional systems that describe the evolution of the parameter distribution with time along with the total population size. Such (and more complex) models were treated in \cite{Ackleh1998,Ackleh1999} from the general point of the theory of differential equations in infinite dimensional spaces. However, it turns out that many such models can be reduced to low dimensional systems of ordinary differential equations (ODEs) \cite{karev2009}. The theory of heterogeneous populations, outlined in \cite{karev2009}, provides the conditions when the mean parameter value $\mathsf{E}_t[\cdot]$ can be effectively calculated at any time moment using only the knowledge of the initial parameter distribution $p(0,\omega)$. Therefore the heterogeneous models, of which \eqref{eq1:5} is the simplest example, can be studies analytically. The examples include the analysis of heterogeneous Malthus equation \eqref{eq1:5} \cite{Karev2003,Karev2005}, Lotka--Volterra system \cite{novozhilov2004} (see also \cite{Ackleh2000}), tumor cell dynamics \cite{karev2006}, and the replicator equation \cite{karev2010}.

A very rich field for the mathematical models with parametric heterogeneity is theoretical epidemiology, where it is natural to assume that individuals vary with respect to their susceptibility to a disease, infectivity to pass a pathogen, contact number, etc. (see, e.g., \cite{Anderson1991,boylan1991note,bnp2010,coutinho1999modelling,Diekmann2000,Diekmann1990,Dwyer2002,Dwyer2000,Dwyer1997,Nikolaou2006,Veliov2005}). For many models from the cited literature, the general approach of the theory of heterogeneous populations yields important analytical results \cite{Novozhilov2008,nov2009hetero}. It is the main goal of the present manuscript to review and generalize these results.

The rest of the paper is organized as follows. In Section \ref{section:2} we present the necessary mathematical facts from the theory of heterogeneous populations with parametric heterogeneity. Section \ref{section:3} is devoted to formulation of various mathematical models of the epidemic spread that include parametric heterogeneity. In Section \ref{section:4} the machinery from Section \ref{section:2} is applied to the models from Section \ref{section:3}. Finally, Section \ref{conclusion} is devoted to discussion of some open problems and conclusions.

\section{Mathematical theory of heterogeneous populations}\label{section:2}
In this section we present the mathematical development of the theory of heterogeneous populations with parametric heterogeneity as it is required for the subsequent application to the mathematical models of the epidemic spread. We aim to present neither the most abstract possible formulations for such models (dubbed \emph{systems with inheritance}, see \cite{gorban2007selection,gorban2010self}) nor the most general form of the dynamical systems that can be analyzed by similar means \cite{karev2009}. We refer the reader to the literature for a notably wider account of the necessary theory \cite{gorban2007selection,karev2009,karev2010replicator}, and remark that our presentation has a primary goal to make the exposition self-contained.

Let us consider two interacting populations, each of which is characterized by its own parameter (trait) value at any given time moment, such that we can speak of the parameter  distribution in each population (the populations are \textit{heterogeneous} with respect to a given parameter). Denote the densities of the populations as $n_1(t,\omega_1)$ and $n_2(t,\omega_2)$. This implies that the total population sizes are given by
$$
N_1(t)=\int_{\Omega_1} n_1(t,\omega_1)\,d\omega_1,\quad N_2(t)=\int_{\Omega_2} n_2(t,\omega_2)\,d\omega_2,
$$
respectively. Here $\Omega_i,\,i=1,2$ are the sets of admissible parameter values, and we suppose that these sets are such that the corresponding integrals are always well defined (in fact, in the examples we usually suppose that $\Omega_i,\,i=1,2$ are intervals of $\mathbb R^1$, but the general theory does not require it). An obvious generalization of the considered situation is that $k$ populations are considered, some of which could be characterized by their own vector-parameters. Such generalization would require only additional notation, and therefore we deal with, without loss of generality, the case of two populations. For these two populations the pdfs are given by
$$
p_i(t,\omega_i)=\frac{n_i(t,\omega_i)}{N_i(t)}\,,\quad i=1,2,
$$
such that $p_i(t,\omega_i)\geq 0,\,\int_{\Omega_i}p_i(t,\omega_i)\,d\omega_i=1,\,i=1,2$ for any $t>0$.

To describe the dynamics of the densities of the populations under study, we start with the general model
\begin{equation}\label{eq2:1}
    \begin{split}
      \frac{\pdt n_1(t,\omega_1)}{n_1(t,\omega_1)} &= F_1\bigl(n_1(t,\omega_1),n_2(t,\omega_2)\bigr), \\
       \frac{\pdt n_2(t,\omega_2)}{n_2(t,\omega_1)} &= F_2\bigl(n_1(t,\omega_1),n_2(t,\omega_2)\bigr),
    \end{split}
\end{equation}
where $F_1,\,F_2$ are given functions (the per capita growth rates of the interacting populations) and $\pdt\equiv \frac{\partial}{\partial t}\,$. For \eqref{eq2:1} the initial conditions are $n_1(0,\omega_1)$ and $n_2(0,\omega_2)$.

It is the form of $F_1$ and $F_2$ that defines the dynamics of interacting populations, and to apply the theory of heterogeneous populations, as presented in, e.g., \cite{karev2009}, this form must satisfy some additional requirements. In particular, we adopt that
\begin{equation}\label{eq2:2}
    F_i\bigl(n_1(t,\omega_1),n_2(t,\omega_2)\bigr)=f_i(\V v)+\varphi(\omega_i)g_i(\V v),\quad i=1,2,
\end{equation}
where $f_i,\,g_i,\,\varphi_i,\,i=1,2$ are given functions,
\begin{equation}\label{eq2:2a}
\V v=(N_1,N_2,\bar{\varphi}_1(t),\bar{\varphi}_2(t)),
\end{equation}
and
\begin{equation}\label{eq2:3}
    \bar{\varphi}_i(t)=\mathsf{E}_t[\varphi_i]=\int_{\Omega_i}\varphi_i(\omega_i)p_i(t,\omega_i)\,d\omega_i,\quad i=1,2.
\end{equation}
In words, we require that the per capita growth rates depend explicitly only on the parameters, total population sizes and the mean values of the functions of the distributed parameters. Assume that we are interested only in the dynamics of the total population sizes. Integrating the first equation in \eqref{eq2:1} with respect to $\omega_1$ and the second one with respect to $\omega_2$, and using \eqref{eq2:2}, \eqref{eq2:3}, we obtain
\begin{equation}\label{eq2:4}
   \dot N_i = N_i\bigl(f_i(\V v)+\bar{\varphi}_i(t)g_i(\V v)\bigr), \quad i=1,2.
\end{equation}
Inasmuch as $\V v$ depends only on the total population sizes and $\bar{\varphi}_i(t),\,i=1,2$, the dynamics of the populations sizes can be found if $\bar{\varphi}_i(t),\,i=1,2$ are known. We remark that $\bar{\varphi}_i(t),\,i=1,2$ depend on the current parameter distributions, which actually has to be found. The remarkable fact is that it is possible to calculate $\bar{\varphi}_i(t),\,i=1,2$ for any time moment if the initial parameter distributions
$$
p_i(0,\omega_i)=\frac{n_i(0,\omega_i)}{N_i(0)}\,,\quad i=1,2
$$
are known and two additional differential equations are allowed. This is the main result for our presentation. To give it a precise statement, introduce the notations
\begin{equation}\label{eq2:5}
    \textsf{M}_i(t,\lambda)=\int_{\Omega_i}e^{\lambda \varphi_i(\omega_i)}p_i(t,\omega_i)\,d\omega_i,\quad i=1,2,
\end{equation}
for the moment generating functions (mgfs) of $\varphi_i(\omega_i),\,i=1,2$ respectively at any time moment. The mgf, it it exists, defines uniquely the given probability distribution. In the mgfs for the initial distributions we occasionally suppress the dependence on time:
$$
\mathsf{M}_i(\lambda)\equiv \mathsf{M}_i(0,\lambda),\quad i=1,2.
$$
It is a basic fact that if $\mathsf{M}_i(t,\lambda),\,i=1,2$ are known then the mean values of $\varphi_i(\omega_i)$ can be found by simple differentiation:
$$
\bar{\varphi}_i(t)=\mathsf{E}_t[\varphi_i]=\left.\partial_\lambda \mathsf{M}_i(t,\lambda)\right|_{\lambda=0},\quad i=1,2.
$$
Now we can state the following main
\begin{theorem}\label{th2:1} Let the dynamics of two interacting populations be described by system \eqref{eq2:1} with \eqref{eq2:2}, \eqref{eq2:2a}, \eqref{eq2:3}. Consider auxiliary variables $q_i(t),\,i=1,2$ that satisfy the following differential equations
\begin{equation}\label{eq2:6}
    \dot q_i=g_i(\V v),\quad q_i(0)=0,\quad i=1,2,
\end{equation}
where $g_i,\,i=1,2$ are as in \eqref{eq2:2}.
Then
\begin{equation}\label{eq2:7}
    \mathsf{M}_i(t,\lambda)=\frac{\mathsf{M}_i(0,\lambda+q_i(t))}{\mathsf{M}_i(0,q_i(t))}=\frac{\mathsf{M}_i(\lambda+q_i(t))}{\mathsf{M}_i(q_i(t))}\,,\quad i=1,2.
\end{equation}
\end{theorem}
\begin{proof}
From the first equation of \eqref{eq2:1} with \eqref{eq2:2} we have
\begin{equation}\label{eq2:8}
    \begin{split}
      n_1(t,\omega_1) & =n_1(0,\omega_1)\exp\{\int_0^t f_1(\V v)\,dt+\varphi_1(\omega_1)\int_0^t g_1(\V v)\,dt\}= \\
        & =n_1(0,\omega_1)\exp\{\int_0^t f_1(\V v)\,dt+\varphi_1(\omega_1)q_1(t)\}.
    \end{split}
\end{equation}
Using \eqref{eq2:8}, we find
\begin{equation}\label{eq2:9}
    \begin{split}
      N_1(t)&=\int_{\Omega_1} n_1(t,\omega_1)\,d\omega_1=\\
       & =\exp\{\int_0^t f_1(\V v)\,dt\} \int_{\Omega_1}n_1(0,\omega_1)\exp\{\varphi_1(\omega_1)q_1(t)\}\,d\omega_1= \\
        & =\frac{1}{N_1(0)}\exp\{\int_0^t f_1(\V v)\,dt\}\int_{\Omega_1}p_1(0,\omega_1)\exp\{\varphi_1(\omega_1)q_1(t)\}\,d\omega_1=\\
        &=\frac{1}{N_1(0)}\exp\{\int_0^t f_1(\V v)\,dt\} \mathsf{M}_1(0,q_1(t)).
    \end{split}
\end{equation}
Using the definition \eqref{eq2:5} for mgf,
\begin{equation}\label{eq2:10}
    \begin{split}
      \mathsf{M}_1(t,\lambda)&=\int_{\Omega_1} \exp\{\lambda \varphi_1(\omega_1)\} p_1(t,\omega_1)\,d\omega_1=\\
       & =\frac{1}{N_1(t)}\int_{\Omega_1} \exp\{\lambda \varphi_1(\omega_1)\} n_1(t,\omega_1)\,d\omega_1=\mbox{from \eqref{eq2:8}}\\
       &=\frac{\int_0^t f_1(\V v)\,dt}{N_1(0)N_1(t)}\int_{\Omega_1}\exp\{(\lambda+q_1(t))\varphi_1(\omega_1)\}p_1(0,\omega_1)\,d\omega_1=\\
       &=\frac{\int_0^t f_1(\V v)\,dt}{N_1(0)N_1(t)}\mathsf{M}_1(0,\lambda+q_1(t)).
    \end{split}
\end{equation}
Finally, applying \eqref{eq2:9} and \eqref{eq2:10} we obtain \eqref{eq2:7}. The same calculations are valid for the second population.
\end{proof}
\begin{remark}Theorem \ref{th2:1} allows to reduce system \eqref{eq2:1} with \eqref{eq2:2} to the four dimensional system of ODEs \eqref{eq2:4}, \eqref{eq2:6}, where the mean parameter values are
\begin{equation}\label{eq2:11}
    \bar{\varphi}_i(t)=\mathsf{E}_t[\varphi_i]=\left.\partial_\lambda \ln \mathsf{M}_i(\lambda)\right|_{\lambda=q_i(t)},\quad i=1,2.
\end{equation}
It is a simple matter (see, e.g., \cite{karev2009}) to prove that the mean parameter values satisfy the equations
\begin{equation}\label{eq2:12}
    \dot{\bar{\varphi}}_i(t)=g_i(\V v)\sigma^2_i(t),\quad i=1,2,
\end{equation}
where $\sigma_i^2(t),\,i=1,2$ are the current variances of $\varphi_i(\omega_i),\,i=1,2$.
\end{remark}

\begin{remark} Using Theorem \ref{th2:1} we find that the solution for the total population size of the heterogeneous Malthus equation \eqref{eq1:5} is given explicitly
as
$$
N(t)=N(0)\mathsf{M}(t),
$$
where $\mathsf{M}(t)$ is the mgf of the initial parameter distribution $p(0,\omega)$. This solution was studied in detail in \cite{Karev2000a,Karev2003}.
\end{remark}
\begin{remark}Theorem \ref{th2:1} shows that the mgfs at any time moment can be found from the mgfs at the initial moment using \eqref{eq2:7}. This theorem also shows that by reducing the original infinite dimensional system \eqref{eq2:1} with \eqref{eq2:2} to the system of ODEs \eqref{eq2:4}, \eqref{eq2:6} we actually do not loose any information because, as far as the functions $q_i(t),\,i=1,2$ are known, the densities $n_i(t,\omega_i),\,i=1,2$ at any time moment can be found from the total population sizes and the time dependent probability distributions $p_i(t,\omega_i),\,i=1,2$, which can be inferred from the corresponding mgfs. For the following exposition we do not require explicit formulas for the evolution of distributions, and refer the reader to \cite{karev2009}, where these results can be found.
\end{remark}
\begin{remark}Theorem \ref{th2:1} provides an analytical proof of \eqref{eq2:7}. It is interesting to note that there is a simple probabilistic proof of a similar formula \cite{aalen1994effects,aalen2008survival}.

Consider the so-called proportional \textit{frailty model} \cite{aalen2008survival}: it is assumed that the hazard rate of an individual is given as the product of an individual specific quantity $Z$ and a basic rate $\alpha(t)$:
\begin{equation}\label{eq2:13}
    \alpha(t|Z)=\alpha(t)Z.
\end{equation}
Here $Z$ plays the role of the parameter distributed in the population. Given $Z$ the probability of surviving up to time $t$ is
$$
S(t|Z)=e^{-ZA(t)},\quad A(t)=\int_0^t\alpha(\tau)\,d\tau.
$$
The population survival function is therefore
\begin{equation}\label{eq2:14}
    S(t)=\mathsf{P}[T>t]=\mathsf{E}[e^{-ZA(t)}]=\mathsf{M}(-A(t)).
\end{equation}
The frailty distribution in the population surviving at time $t$ can be found as follows (here $\mathsf{I}(T>t)$ is the indicator function):
\begin{equation}\label{eq2:15}
    \begin{split}
      \mathsf{M}(t,\lambda) & =\mathsf{E}_t[e^{\lambda Z}]=\mathsf{E}[e^{\lambda Z}|T>t]=\frac{\mathsf{E}[e^{\lambda Z}\mathsf{I}(T>t)]}{\mathsf{P}[T>t]}= \\
        & =\frac{\mathsf{E}[e^{\lambda Z-ZA(t)}]}{\mathsf{E}[e^{-ZA(t)}]}=\frac{\mathsf{M}(\lambda-A(t))}{\mathsf{M}(-A(t))}\,.
\end{split}
\end{equation}
Using $q(t)=-A(t)$ in \eqref{eq2:15} we obtain \eqref{eq2:7}.

\end{remark}

Concluding, using the mathematical theory of heterogeneous populations with parametric heterogeneity, outlined in this section, it is possible to model communities of populations when each population is characterized by its own parameter value, and the per capita growth rates depend on this parameter and the average characteristics of the interacting populations.

\section{Heterogeneous models in epidemiology: Model formulation}\label{section:3}
The modern mathematical epidemiology has its roots in now classical work by Kermack and McKendrick \cite{Diekmann1993,Kermack1927}, where the total population of the constant size $N$ was subdivided into three categories: susceptible individuals $S$ that are prone to infection, infectious individuals $I$ that transmit the disease, and removed individuals $R$ that either acquire life-long immunity or die. In the simplest case, assuming that the period of infection is exponentially distributed with the mean $1/\gamma$ and the transmission process is described by mass-action kinetics \cite{Heesterbeek2005,mccallum2001} (this means that the contact rate is proportional to the total population size $N$), we have that the dynamics is described by
\begin{equation}\label{eq3:1}
    \begin{split}
      \dot{S} & =-\beta SI, \\
        \dot{I} &=\beta SI-\gamma I,\\
        \dot{R}&=\gamma I,
        \end{split}
\end{equation}
where $\beta$ is the transmission parameter that encompasses the information on  the probability of a successful contact (i.e., the contact that results in infection) and the contact rate. From \eqref{eq3:1} it follows that $S(t)+I(t)+R(t)=N$ for any $t$, therefore the third equation is redundant and in the following we usually omit it. SIR model \eqref{eq3:1} includes two parameters that in reality vary from individual to individual. Therefore, it is important to take this heterogeneity into account.

\subsection{Heterogeneous SIR model with distributed susceptibility to a disease}
There are a number of studies in the literature which model heterogeneous susceptibility to a disease, with either a finite number of different susceptibility classes  \cite{andersson2000stochastic,ball1985deterministic,bonzi2011stability,Dushoff1999,gart1972statistical,Hsu2002,Hyman2005,rodrigues2009heterogeneity,scalia1986asymptotic} or with a continuous distribution of susceptibility \cite{boylan1991note,coutinho1999modelling,Dwyer2002,Dwyer2000,Dwyer1997}. We present a formulation from \cite{Novozhilov2008} keeping in mind that both discrete and continuous distributions can be accommodated.

Let us denote $s(t,\omega)$ the density of the susceptible individuals having the trait value $\omega$. The total size of the susceptible population is $S(t)=\int_\Omega s(t,\omega)\,d\omega$. Assuming the law of mass action, we obtain that the changes in the susceptible and infectious populations are described by
\begin{equation}\label{eq3:2}
    \begin{split}
      \pdt s(t,\omega) & =-\beta(\omega)s(t,\omega)I(t), \\
       \dot{I}(t) &= I(t)\int_\Omega \beta(\omega)s(t,\omega)\,d\omega-\gamma I(t).
    \end{split}
\end{equation}
Or, using the notations $$\bar{\beta}(t)=\mathsf{E}_t[\beta]=\int_\Omega \beta(\omega)p_s(t,\omega)\,d\omega,\quad p_s(t,\omega)=\frac{s(t,\omega)}{S(t)}\,,$$ \eqref{eq3:2} can be rewritten as
\begin{equation}\label{eq3:3}
    \begin{split}
      \pdt s(t,\omega) & =-\beta(\omega)s(t,\omega)I(t), \\
       \dot{I}(t) &= \bar{\beta}(t)SI-\gamma I(t).
    \end{split}
\end{equation}
For \eqref{eq3:3} the initial conditions are $s(0,\omega)=s_0(\omega)=S_0p_s(0,\omega),\,I(0)=I_0$.

Model \eqref{eq3:2}, \eqref{eq3:3} can be also obtained from the general epidemic equation \cite{Diekmann2000}
\begin{equation}\label{eq3:4}
    \pdt s(t,\omega)=s(t,\omega)\int_\Omega\int_0^\infty A(\tau,\omega,\eta)\pdt s(t-\tau,\eta)\,d\tau d\eta,
\end{equation}
where $A(\tau,\omega,\eta)$ is the expected infectivity of an individual that was infected $\tau$ units ago while having trait value $\eta$ towards to a susceptible with trait value $\omega$. If we assume that $A(\tau,\omega,\eta)=\beta(\omega)f(\tau)$ for a given function $f$ and set
$$
I(t)=-\int_\Omega\int_0^t f(\tau)\pdt s(t-\tau,\eta)\,d\tau d\eta,
$$
then, after some algebra, we obtain
\begin{equation}\label{eq3:5}
    \begin{split}
      \pdt s(t,\omega) & =-\beta(\omega)s(t,\omega)I(t), \\
       \dot{I}(t) &= f(0)I(t)\int_\Omega \beta(\omega)s(t,\omega)\,d\omega-\int_\Omega\int_0^\infty f'(\tau)\pdt s(t-\tau,\eta)\,d\tau d\eta.
    \end{split}
\end{equation}
Letting $f(\tau)=e^{-\gamma \tau}$ reduces \eqref{eq3:5} to \eqref{eq3:3} (this can be shown to be the only case to end up with an ODE system, for this it is necessary and sufficient that $f$ satisfies the equation $\dot f=-\gamma f$).
\begin{remark} If in the general epidemic equation \eqref{eq3:4} it is assumed that $f(\tau)=\chi(T-\tau)$, where $\chi$ is the Heaviside function then from \eqref{eq3:5} it follows that
\begin{equation}\label{eq3:6}
    \begin{split}
      \pdt s(t,\omega) & =-\beta(\omega)s(t,\omega)I(t), \\
       \dot{I}(t) &= \bar{\beta}(t)SI-\bar{\beta}(t-T)S(t-T)I(t-T).
    \end{split}
\end{equation}
Model \eqref{eq3:6} is the model with distributed susceptibility studied in \cite{Dwyer2002,Dwyer2000,Dwyer1997}.
\end{remark}

\subsection{Heterogeneous SIR model with distributed infectivity}
Let $\beta(\omega)$ be the transmission parameter of an individual with the infectivity value $\omega$, and $i(t,\omega)$ be the density of the infectious hosts with trait value $\omega$ at the time moment $t$, $I(t)=\int_\Omega i(t,\omega)\,d\omega$. We assume that the susceptible population is homogeneous. To describe the dynamics of the infectious population it is necessary to specify what trait value is assigned to a newly infected individual, which was infected by an infectious individual with the trait value $\eta$. Denoting $\psi(\omega,\eta)$ the pdf that prescribes the probability that a newly infected individual is assigned the trait value $\omega$ if infected by an individual with the trait value $\eta$, we obtain
\begin{equation}\label{eq3:7}
    \begin{split}
      \dot{S}(t) & =-\bar{\beta}(t)S(t)I(t), \\
       \pdt i(t,\omega) &= S(t)\int_\Omega \psi(\omega,\eta)\beta(\eta)i(t,\eta)\,d\eta-\gamma i(t,\omega),
    \end{split}
\end{equation}
where now $$\bar{\beta}(t)=\int_\Omega \beta(\omega) p_i(t,\omega)\,d\omega,\quad p_i(t,\omega)=\frac{i(t,\omega)}{I(t)}\,.$$ The initial conditions are $S(0)=S_0,\,i(0,\omega)=i_0(\omega)=I_0p_i(0,\omega)$.

There are several natural choices for $\psi(\omega,\eta)$, the simplest of which is $\psi(\omega,\eta)=\delta(\omega-\omega')$, where $\delta$ is the delta-function. This option means that the newly infected individual acquires the trait value of the person by whom he was infected (this is equivalent to $\psi(\omega,\eta)=p_i(t,\omega)$, i.e., the trait values are assigned according to the current distribution of the infectivity). Using $\psi(\omega,\eta)=\delta(\omega-\omega')$ in \eqref{eq3:7} we obtain
\begin{equation}\label{eq3:8}
    \begin{split}
      \dot{S}(t) & =-\bar{\beta}(t)S(t)I(t), \\
       \pdt i(t,\omega) &= \beta(\omega)i(t,\omega)S(t)-\gamma i(t,\omega).
    \end{split}
\end{equation}
Model \eqref{eq3:8} is very similar in form to \eqref{eq3:3}. However, it should be clear that, in general, model \eqref{eq3:3} is much closer to reality than \eqref{eq3:8}, which corresponds to the case when several different strains of an infection can be passed on.

\subsection{Heterogeneous SIR model with distributed susceptibility and infectivity}
Let us assume now that both the susceptibility and infectivity are distributed in the population experiencing the disease; this will generalize models \eqref{eq3:3} and \eqref{eq3:8}.

Let $s(t,\omega_1)$ and $i(t,\omega_2)$ be the densities of the susceptible and infectious individuals respectively. For the following, simplifying, we assume that
the traits of the two subpopulations are independent, i.e.,
$\beta(\omega_1,\omega_2)=\beta_1(\omega_1)\beta_2(\omega_2)$. The
number of susceptible hosts with the trait value $\omega_1$ infected by
infectious individuals with the trait value $\omega_2$ is given by
$$\beta_1(\omega_1)s(t,\omega_1)\beta_2(\omega_2)i(t,\omega_2).$$ Therefore,
the total change in the infectious subpopulation with trait value $\omega_2$
is $$\beta_2(\omega_2)i(t,\omega_2)\int_{\Omega_1}\beta_1(\omega_1)s(t,\omega_1)\,d\omega_1,$$ assuming that $\psi(\omega,\eta)=\delta(\omega-\eta)$ (cf. \eqref{eq3:8}).
A similar expression describes the change in the susceptible
population. It is worth emphasizing that nothing else except for the standard
law of mass action is supposed to formulate the terms for the change
in susceptible and infectious subpopulations. Combining the above
assumptions we obtain the following model:
\begin{equation}\label{si1}
    \begin{split}
    \pdt s(t,\omega_1)  &=-\beta_1(\omega_1)s(t,\omega_1)\int_{\Omega_2}\beta_2(\omega_2) i(t,\omega_2)\, d\omega_2\\
           &=-\beta_1(\omega_1)s(t,\omega_1)\bar{\beta}_2(t)I(t),\\
    \pdt i(t,\omega_2)  &=\beta_2(\omega_2)i(t,\omega_2)\int_{\Omega_1}\beta_1(\omega_1) s(t,\omega_1)\,
    d\omega_1-\gamma i(t,\omega_2)\\
&=\beta_2(\omega_2)i(t,\omega_2)\bar{\beta}_1(t)S(t)-\gamma i(t,\omega_2),
\end{split}
\end{equation}
where
$$
\bar{\beta}_1(t)=\mathsf{E}_t[\beta_1]=\int_{\Omega_1}\beta_1(\omega_1)p_s(t,\omega_1)\,d\omega_1,\quad \bar{\beta}_2(t)=\mathsf{E}_t[\beta_2]=\int_{\Omega_2}\beta_2(\omega_2)p_i(t,\omega_2)\,d\omega_2.
$$
Model \eqref{si1} is supplemented with the initial conditions
$s(0,\omega_1)=S_0p_s(0,\omega_1),\,i(0,\omega_2)=I_0p_i(0,\omega_2)$.

In the case $\gamma=0$ we obtain the heterogeneous SI model with distributed susceptibility and infectivity
\begin{equation}\label{si2}
    \begin{split}
    \pdt s(t,\omega_1)  &=-\beta_1(\omega_1)s(t,\omega_1)\bar{\beta}_2(t)I(t),\\
    \pdt i(t,\omega_2)  &=\beta_2(\omega_2)i(t,\omega_2)\bar{\beta}_1(t)S(t),
\end{split}
\end{equation}
for which the global dynamics is simple and is similar to the
simplest homogeneous SI model, when $S(t)\to 0,\,I(t)\to N,$ when $t\to\infty$.

\subsection{Heterogeneous SI model with heterogeneous contact structure}
Above we discussed mainly about heterogeneity of
the hosts: whether all susceptible individuals are of the same type
with equal susceptibility, and whether all infectious individuals
 have equal ability to infect others. Another aspect of heterogeneity
is the possible heterogeneous social contact network, which is the one of the central topics in mathematical epidemiology, e.g,~\cite{andersson2000stochastic,Bansal2007,Danon2011,keeling2008modeling,meyers2007contact}. It is difficult to apply the general theory of
heterogeneous populations as presented in Section \ref{section:2} to such models, however, there
is a simple case, for which some results can be obtained.

Let us assume that $n(t,\omega)$ denotes the density of individuals
in the population, which are making $\omega$ contacts on average.
Every individual can be contacted by another individual, and the individuals
differ in an average number of contacts. This
situation is usually termed as \textit{separable mixing}. If we
denote $r$ the probability of transmission the disease given a
contact, then, the simplest SIR-model with separable mixing can be
described by the following system:
\begin{equation}\label{eq3a:4}
\begin{split}
\pdt s(t,w)&=-r\omega s(t,\omega)\frac{\int_{\Omega}\omega i(t,\omega)d\omega}{\int_{\Omega}\omega s(t,\omega)\,d\omega+\int_{\Omega}\omega i(t,\omega)\,d\omega}\,,\\
\pdt i(t,w)&=-r\omega
s(t,\omega)\frac{\int_{\Omega}\omega
i(t,\omega)d\omega}{\int_{\Omega}\omega s(t,\omega)\,d\omega+\int_{\Omega}\omega i(t,\omega)\,d\omega}-\gamma i(t,\omega)\,.
\end{split}
\end{equation}

In case of SI model \eqref{si2} ($\gamma=0$) we have $s(t,\omega)+i(t,\omega)=n_0(\omega)$ for any $t$, and
$n_0(\omega)$ is a given density which specifies probability density
function of the contact distribution. Using the property
$i(t,\omega)=n_0(\omega)-s(t,\omega)$, we obtain
\begin{equation}\label{eq3a:5}
\pdt s(t,w)=-r\omega
s(t,\omega)\left[1-\frac{\int_{\Omega}\omega
s(t,\omega)d\omega}{\int_{\Omega}\omega n_0(\omega)d\omega}\right].
\end{equation}

We conclude this section with an obvious statement that the list of possible epidemiological models with parametric heterogeneity can be made longer. What is important, however, that models \eqref{eq3:3}, \eqref{eq3:6}, \eqref{eq3:8}, \eqref{si1}, \eqref{si2}, and \eqref{eq3a:5} all are written in the form \eqref{eq2:1}--\eqref{eq2:3}, which allows a unified treatment of these models within the mathematical framework outlined in Section \ref{section:2}.

\section{Mathematical analysis of an SIR model}\label{section:4}
\subsection{Analysis of heterogeneous SIR model with distributed susceptibility and infectivity}
Model \eqref{si1} comprises, as particular cases, models \eqref{eq3:3} and \eqref{eq3:8}. Therefore, the main result is stated for this model (the case of distributed susceptibility alone was treated in \cite{Novozhilov2008}).
\begin{theorem}\label{th:2} Heterogeneous SIR model \eqref{si1} with distributed susceptibility and infectivity is equivalent to the following two-dimensional non-autonomous system of ODEs:
\begin{equation}\label{eq4:1}
    \begin{split}
      \dot{S} &=-h_1(S)h_2(t,I), \\
      \dot{I} &= h_1(S)h_2(t,I)-\gamma I,
    \end{split}
\end{equation}
where
\begin{equation}\label{eq4:2}
    h_1(S)=S_0({\partial_\lambda \mathsf{M}^{-1}_1(0,\lambda)|_{\lambda=S/S_0}})^{-1},
\end{equation}
\begin{equation}\label{eq4:3}
    h_2(t,I)=I_0e^{-\gamma t}({\partial_\lambda \mathsf{M}^{-1}_2(0,\lambda)|_{\lambda=Ie^{\gamma t}/I_0}})^{-1},
\end{equation}
and $\mathsf{M}^{-1}_i(0,\lambda),\,i=1,2$ are the inverse functions to the mgfs of the initial distributions of susceptibility and infectivity respectively.
\end{theorem}
\begin{proof}According to Theorem \ref{th2:1}, system \eqref{si1} is equivalent to
\begin{equation}\label{eq4:4}
    \begin{split}
      \dot{S} &=-\bar{\beta}_1(t)\bar{\beta}_2(t)SI, \\
      \dot{I} &=\bar{\beta}_1(t)\bar{\beta}_2(t)SI-\gamma I,\\
      \dot q_1 &=-\bar{\beta}_2(t)I,\\
      \dot q_2 &=\bar{\beta}_1(t)S,
    \end{split}
\end{equation}
where
$$
\bar{\beta}_i(t)=\partial_\lambda \ln \mathsf{M}_i(0,\lambda)|_{\lambda=q_i(t)},\quad i=1,2.
$$
From \eqref{eq4:4} it follows that
\begin{equation*}
    \begin{split}
      \frac{d}{dt}{\ln S} &=\frac{d}{dt} \ln \mathsf{M}_1(0,q_1),\\
       \frac{d}{dt} \ln I&= \frac{d}{dt} \ln \mathsf{M}_2 (0,q_2)-\gamma,\\
    \end{split}
\end{equation*}
from which we have two first integrals
\begin{equation}\label{eq4:6}
      \frac{S}{S_0} = \mathsf{M}_1(0,q_1),\quad  \frac{Ie^{\gamma t}}{I_0} = \mathsf{M}_2 (0,q_2).\\
\end{equation}
Since the mgfs are monotone functions, we can express $q_1$ and $q_2$ in \eqref{eq4:6} through $S$ and $I$ respectively. Putting the resulting expressions into \eqref{eq4:4} and remembering the inverse function theorem, we obtain \eqref{eq4:1} with \eqref{eq4:2}, \eqref{eq4:3}.
\end{proof}

Theorem \ref{th:2} gives an important example of a bottom up approach in building mathematical models in epidemiology \cite{de2005unstructured,hastings2005unstructured}. We start with a detailed model when each individual in the population has its own parameter value. This model, albeit still oversimplified and unrealistic, takes into account physiological structure of the population without any \textit{ad hoc} assumptions on the dynamics of the total population sizes. The general theory of the heterogeneous populations with parametric heterogeneity allows to reduce the model to the population-level description, which, incidentally, still can be described as a classical SIR model, which contains non-linear and time-dependent transmission rates. This shows that a great deal of mathematical models, built from the first principles, are still valid and can be used for ecological predictions, because they can be shown to follow from detailed individual-based models.

Let us consider an explicit example (more details are given in \cite{nov2009hetero}).

\paragraph{On the power law transmission function.} The \textit{transmission function} (the number of new infectious cases per time unit) is considered to be one of the major ingredients of the models describing spread of an infectious disease \cite{Diekmann2000}. Historically, borrowing an analogy from the chemical kinetics, this function was supposed to be of a simple bilinear form, $\propto SI$, which means that the law of mass action is assumed to hold \cite{Heesterbeek2005} (this is often called density dependent transmission function). If one assumes that the number of contacts is fixed for any individual and does not depend on the population size, than the transmission function takes the form $\propto SI/N$ (frequency dependent transmission function). If the population size constant these two transmission function yield the same predictions, whereas variable population size can produce different behaviors (e.g., \cite{Berezovskaya2007b,Novozhilov2006a}). Other transmission functions are also possible, see \cite{mccallum2001}. Important point here is that most of the used nonlinear functions, that are somewhat intermediate between density and frequency dependent modes of transmission, are phenomenological and lack any mechanistic derivation.

One of most frequently used transmission function takes the form
\begin{equation}\label{eq4:7}
    T(S,I)=\beta S^p I^q, \quad p,q>0.
\end{equation}
This function is often called the power law transmission function. It was used in, e.g., \cite{wilson1945lma,wilson1945lmab} in the form $T(S,I)=\beta S^p I$. Severo \cite{severo1967} considered the full form \eqref{eq4:7}, see also \cite{Liu1987,Liu1986a} for a detailed mathematical analysis of epidemiological models with the power law transmission function. Simulations \cite{roy2006,Stroud2006} show that power law transmission function indeed improves an accuracy of the mean-field SIR models. Using Theorem \ref{th:2} we can show that the power law transmission function can be inferred from the heterogeneous SIR models \eqref{si1} and \eqref{si2}.

Let us assume that the initial distribution of susceptibility is a gamma-distribution with parameters $k_1$ and $\nu_1$:
\begin{equation}\label{eq4:8}
p_1(0,\omega)=\frac{\nu_1^{k_1}}{\Gamma(k_1)}\omega^{k_1-1}w^{-\nu_1 \omega},\quad \omega\geq 0,\,\quad k_1,\nu_1>0.
\end{equation}
From \eqref{eq4:8} one has, assuming additionally that $\beta_1(\omega_1)=\omega_1$,
\begin{equation}\label{eq4:9}
    \mathsf{M}_1(\lambda)=\left(1-\frac {\lambda}{\nu_1}\right)^{-k_1},
\end{equation}
therefore the expression in \eqref{eq4:2} is given by
\begin{equation}\label{eq4:10}
    h_1(S)=\frac{k_1}{\nu_1}S\left(\frac{S}{S_0}\right)^{1/k_1}.
\end{equation}
An analogous expression can be obtained for $h_2(t,I)$ if it is postulated that the initial distribution of infectivity is also a gamma-distribution \eqref{eq4:8} with parameters $k_2,\nu_2$ and $\beta_2(\omega_2)=\omega_2$. Expression \eqref{eq4:10} implies
\begin{corollary} The power law transmission function \eqref{eq4:7} with the heterogeneity parameters $q=1,\,p=1+1/k_1$ can be obtained as a consequence of the distributed heterogeneous SIR model \eqref{eq3:3} with distributed susceptibility if the initial distribution of susceptibility is the gamma-distribution with parameters $k_1,\,\nu_1$.

The power law transmission function \eqref{eq4:7} with the heterogeneity parameters $q=1+1/k_2,\,p=1+1/k_1$ can be obtained as a consequence of the distributed heterogeneous SI model \eqref{si2} with distributed susceptibility and infectivity if the initial distribution of susceptibility is the gamma-distribution with parameters $k_1,\,\nu_1$ and the initial distribution of infectivity if the gamma-distribution with parameters $k_2,\nu_2$.
\end{corollary}
\begin{remark}\label{remark:9}Using \eqref{eq2:7} and \eqref{eq2:11} it is straightforward to show that for the initial gamma distribution of susceptibility, the time-dependent distribution is also the gamma-distribution with parameters $k_1$ and $\nu_1-q_1(t)$, where $q_1(t)$ is the solution of the corresponding auxiliary equation \eqref{eq2:6}. That is, we have
$$
\bar{\beta}_1(t)=\mathsf{E}_t[\beta_1]=\frac{k_1}{\nu_1-q_1(t)}\,,\quad \sigma_1^2(t)=\frac{k_1}{(\nu_1-q_1(t))^2}\,.
$$
Note that for any time moment the coefficient of variation $cv=\sigma_1(t)/\bar{\beta_1}(t)=1/\sqrt{k_1}$ and does not depend on $t$.
\end{remark}

\begin{remark} To analytically analyze models with parameter distributions it is useful to have families of distributions for which their mgfs are known, as in the case of the gamma-distribution. A very general family of distributions is the so-called power variance function distributions (PVF distributions) \cite{aalen2008survival}, which are defined through mgf:
\begin{equation}\label{eq4:11}
    \mathsf{M}(\lambda)=\exp\left[-\rho\left\{1-\left(\frac{\nu}{\nu-\lambda}\right)^m\right\}\right],
\end{equation}
with $\nu>0,\,m>-1,\,m\rho>0$.

The mgf \eqref{eq4:11} describes a number of distributions. For instance, if $\rho\to\infty$ and $m\to 0$ in such a way that $\rho m\to k$, then \eqref{eq4:11} reduces to the mgf of the gamma-distribution \eqref{eq4:9}. When $m>0$ then \eqref{eq4:11} gives a compound Poisson distribution, which is the sum of independent gamma-distributions with the parameters $\nu$ and $m$, when the number of summand is Poisson distributed with expectation $\rho$. Such compound Poisson distribution has non-zero mass probability at zero, which means, in terms of susceptibility to a disease, that there is part of population of mass $\exp\{-\rho\}$ that is immune to the disease. When $m=-1/2$ and $\rho<0$ we have an inverse gaussian (Wald) distribution. Other limiting distributions are possible, which gives a wide choice of heterogeneity distributions \cite{aalen2008survival}.

From \eqref{eq4:11} it can be shown that
$$
\mathsf{E}[Z]=\frac{\rho m}{\nu}\,,\quad \mathsf{Var}[Z]=\frac{\rho m}{\nu}\frac{m+1}{\nu}\,.
$$
Equation \eqref{eq2:7} implies that the time dependent parameter distribution of the models with parametric heterogeneity is still PVF distribution if the initial one is given by \eqref{eq4:11}, with the following change in parameters:
$$
\rho\to\rho\left(\frac{\nu}{\nu-q(t)}\right)^m,\quad \nu\to\nu-q(t),\quad m\to m.
$$
Not all the distributions possess this ``stability'' property: e.g., the initial uniform distribution turns into truncated exponential distribution.
\end{remark}
\paragraph{On the final epidemic size.} A very important quantity of the SIR model \eqref{eq3:1} is the final epidemic size, which can be defined as the number of susceptible hosts that escape infection, and which we denote $S_{\infty}$. From \eqref{eq3:1} it follows that in the homogeneous case this number can be found as the root to the equation
\begin{equation}\label{eq4:12}
    S_\infty=S_0e^{{\beta(S_\infty-N)}/{\gamma}}.
\end{equation}
Consider again the heterogeneous SIR model \eqref{eq3:3} with distributed susceptibility. Theorem \ref{th:2} implies
\begin{corollary}The final epidemic size of the heterogeneous SIR model with distributed susceptibility \eqref{eq3:3} can be found as the root to the equation
\begin{equation}\label{eq4:13}
    S_\infty=S_0\mathsf{M}\bigl((S_\infty-N)/\gamma\bigr).
\end{equation}
\end{corollary}
\begin{proof}
From the first equation in \eqref{eq4:4} and recalling the third equation in \eqref{eq3:1} we find that
$$
\frac{d}{dt}\ln S(t)=\frac{d}{dt}\ln \mathsf{M}(-R/\gamma),
$$
from which, after integration and using the fact $R_\infty=N-S_\infty$, \eqref{eq4:13} follows.
\end{proof}
\begin{remark} Equation \eqref{eq4:12} is a particular case of \eqref{eq4:13} if the initial distribution of susceptibility is the delta-function.
\end{remark}

\begin{remark} Equation \eqref{eq4:13} can also be obtained from the general epidemic equation \eqref{eq3:4}, which shows that no assumptions on the distribution of epidemic length is required to deduce \eqref{eq4:13}. A very careful mathematical analysis of this equation is given in \cite{katriel2011size}, and we refer the reader to this reference for many intricate details. We note that for the first time equation \eqref{eq4:13} was written in \cite{ball1985deterministic} for a discrete distribution of susceptibility. It can be easily proved that the final epidemic size equation implies that the epidemic is the most severe in case of a totally uniform population; heterogeneity in susceptibility increases the final epidemic size $S_\infty$ (see Fig. \ref{fig:1}).
\end{remark}

Consider a numerical example. Let us assume that we have two SIR model \eqref{eq3:3} with heterogeneous susceptibility, and the initial distributions are given by a gamma and Wald distributions respectively with the same initial mean values and the same variances. The final epidemic sizes are shown in Fig. \ref{fig:1} depending on the initial variances.
\begin{figure}[!tb]
\centering
\includegraphics[width=0.7\textwidth]{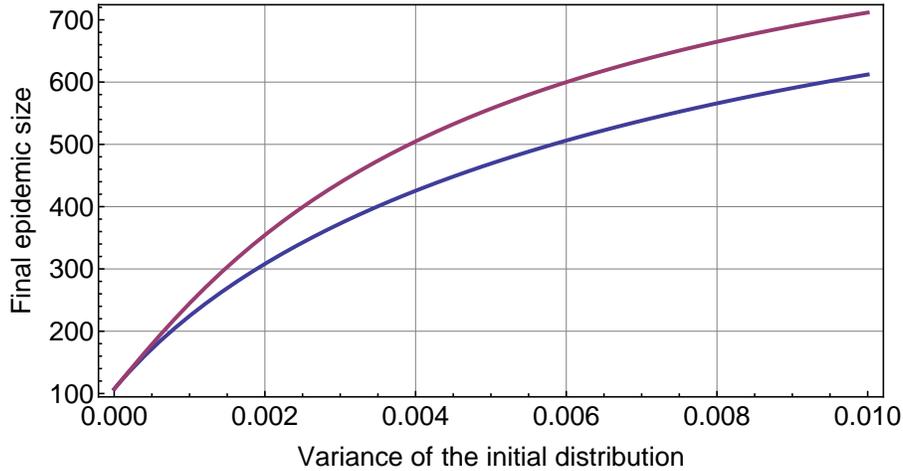}
\caption{Final epidemic sizes versus the initial variance of the susceptibility distributions for two initial distributions: a gamma distribution (the top curve) and a Wald distribution (the bottom curve). All other parameters are the same for both cases}\label{fig:1}
\end{figure}

Figure \ref{fig:1} allows to make an important conclusion: to infer the final epidemic size for an SIR model with distributed susceptibility it is not enough to know only several first moments of the initial distribution. This conclusion is somewhat restrictive for a predictive use of such models; however, it also signifies that various approximation techniques can lead to erroneous conclusions. As an example consider the final epidemic size equation for \eqref{eq3:6}, obtained in \cite{Dwyer1997} with two different methods. First it was supposed that the susceptibility distribution is a gamma-distribution. The second approach was to consider an infinite dimensional system of ODE, which can be inferred from \eqref{eq3:6}, for the moments of the corresponding distribution (this is a natural strategy for analyzing infinite dimensional dynamical systems of the form \eqref{eq2:2}, see \cite{Nikolaou2006,Veliov2005}). Inasmuch as it is impossible to solve a system with infinite number of equations, various techniques to close such systems exit. In particular, it was supposed in \cite{Dwyer1997} that the coefficient of variation is constant. This assumption led to the same result as was obtained for the initial gamma distribution. Therefore, it was concluded that the exact form of the initial distribution is irrelevant because two seemingly different approaches lead to the same outcome. However, theory from Section \ref{section:2} shows that the opposite is true. First, Remark \ref{remark:9} gives the explanation why two approaches in \cite{Dwyer1997} turned out to be equivalent. And second, Fig. \ref{fig:1} provides indirect proof that the final epidemic size depends on the exact form of the underlying susceptibility distribution (see also \cite{katriel2011size} for exact mathematical statements).

Unfortunately, for the general case of the SIR model \eqref{si1} with distributed susceptibility and infectivity we were not able to obtain a simple equation for the final epidemic size.

\paragraph{Numerical illustration.} In Fig. \ref{fig:2} numerical solutions for \eqref{eq3:3} are shown for different initial variances of the parameter distribution and equal means, which confirms that the heterogeneity in susceptibility decreases the severity an an epidemic.
\begin{figure}
\centering
\includegraphics[width=0.48\textwidth]{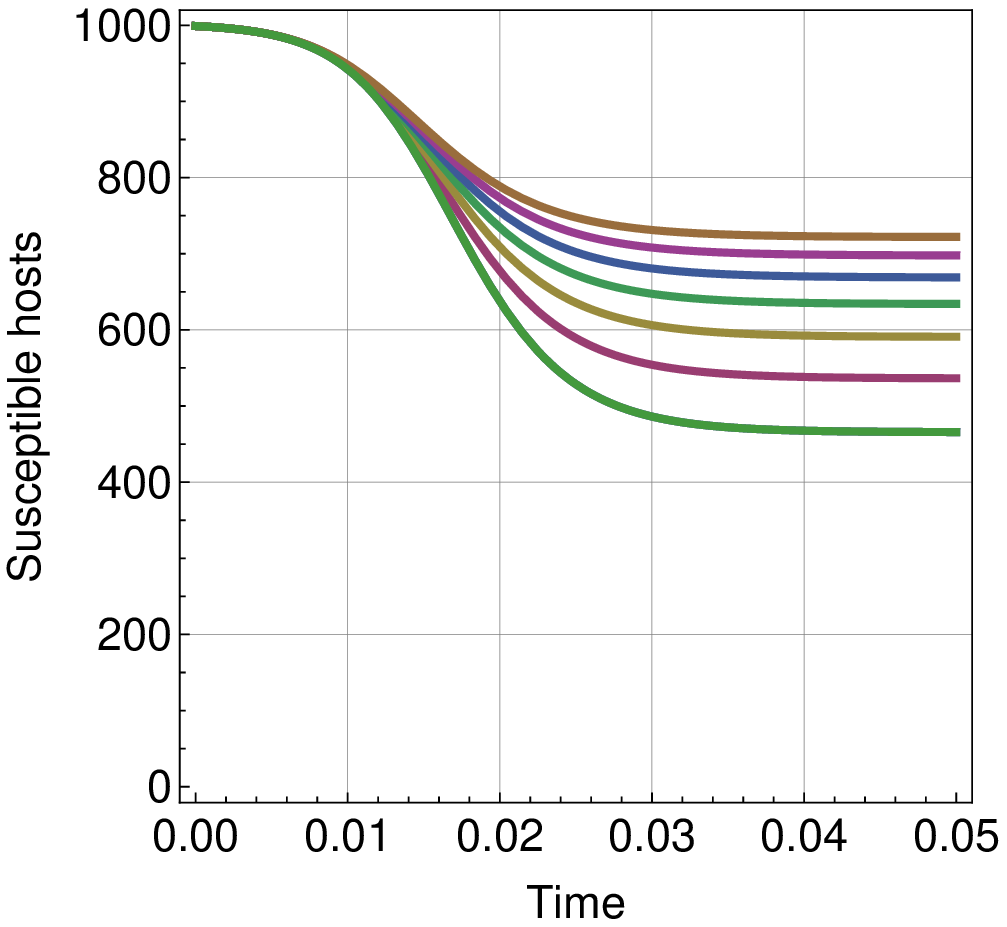}
\includegraphics[width=0.48\textwidth]{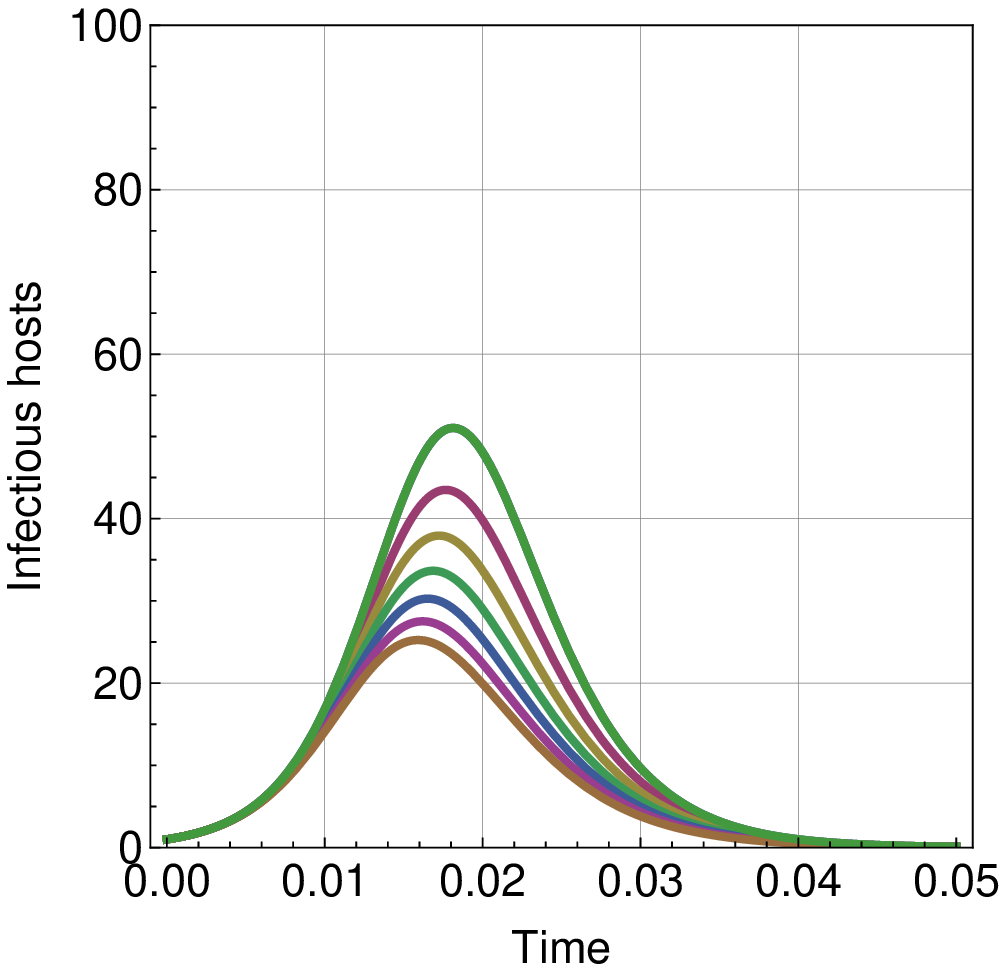}
\caption{Solutions to the heterogeneous SIR model \eqref{eq3:3} with distributed susceptibility. The parameters are $S_0=999,\,I_0=1,\,\gamma=700,\,\bar{\beta_1}(0)=1$ and the initial distribution of susceptibility is the gamma distribution with $\sigma^2_1(0)=0,0.2,0.4,0.6,0.8,1,1.2$ (from the bottom to top curves for the susceptible hosts, and the opposite direction for the infectious hosts)}\label{fig:2}
\end{figure}

In Fig. \ref{fig:3} numerical solutions for system \eqref{eq3:8} are shown for different initial variances of infectivity and equal means. This figure shows that the distributed infectivity has an opposite effect on the severity of an epidemic comparing with the distributed susceptibility: the more heterogeneous the initial distribution of infectivity is, the smaller the number of susceptible hosts that escape the infection.
\begin{figure}[!t]
\centering
\includegraphics[width=0.48\textwidth]{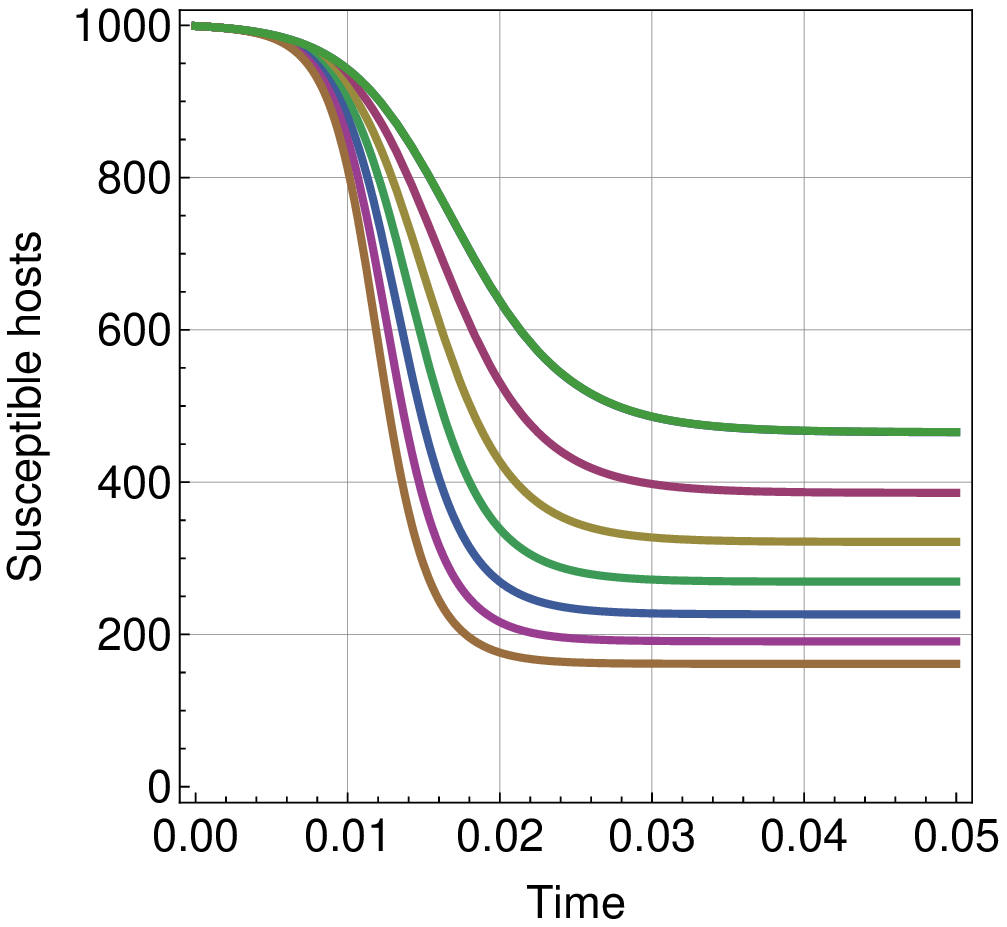}
\includegraphics[width=0.48\textwidth]{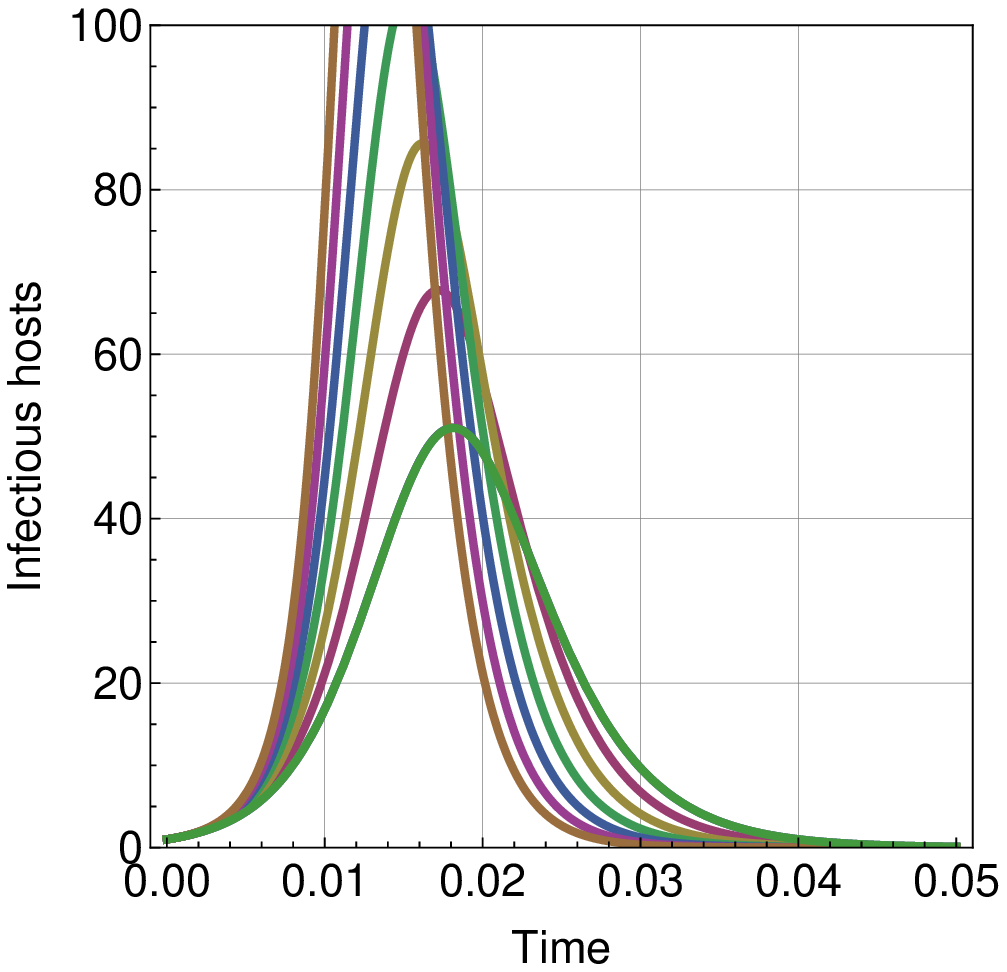}
\caption{Solutions to the heterogeneous SIR model \eqref{eq3:8} with distributed infectivity. The parameters are $S_0=999,\,I_0=1,\,\gamma=700,\,\bar{\beta_1}(0)=1$ and the initial distribution of infectivity is the gamma distribution with $\sigma^2_2(0)=0,0.005,0.01,0.015,0.02,0.025,0.03$ (from the top to bottom curves for the susceptible hosts, and the opposite direction for the infectious hosts)}\label{fig:3}
\end{figure}

The last conclusion can be cast in an exact mathematical statement if one is interested only in the initial period of the epidemic.
\begin{corollary}\label{corr:14}
Let $S_1(t),\,S_2(t)$ be the solutions to \eqref{eq3:3} with the initial conditions that satisfy $\sigma_1^2(0)>\sigma_2^2(0)$ for the distribution of susceptibility, all other initial conditions being equal. Then there exists $\varepsilon>0$ such that $S_1(t)>S_2(t)$ for all $t\in(0,\varepsilon)$.

Let $S_1(t),\,S_2(t)$ be the solutions to \eqref{eq3:8} with the initial conditions that satisfy $\sigma_1^2(0)>\sigma_2^2(0)$ for the distribution of infectivity, all other initial conditions being equal. Then there exists $\varepsilon>0$ such that $S_1(t)<S_2(t)$ for all $t\in(0,\varepsilon)$.
\end{corollary}
\begin{proof} Using \eqref{eq2:12} and differentiating the first equation in \eqref{eq3:3} we obtain
$$
S''(t)=I^2S(\sigma^2(t)+\bar{\beta}(t))-\bar{\beta}(t)I'S,
$$
or, at the initial time moment, $S_1''(0)>S''_2(0)$, which proves the first part the corollary. The second part is proved in a similar way.
\end{proof}

Finally, for the general heterogeneous SIR model \eqref{si1} with distributed infectivity and susceptibility some numerical results are shown in Fig. \ref{fig:4}, where the initial distributions of susceptibility and infectivity are combined from those in Figs. \ref{fig:2} and \ref{fig:3}. The conclusion from numerical calculations is that the interaction of the distributions of susceptibility and infectivity is nonlinear and cannot be predicted from knowing the final outcomes of separate epidemics  (see also \cite{andreasen2011final} for the discussion of the final epidemic size of the model \eqref{si1}).
\begin{figure}[!t]
\centering
\includegraphics[width=0.48\textwidth]{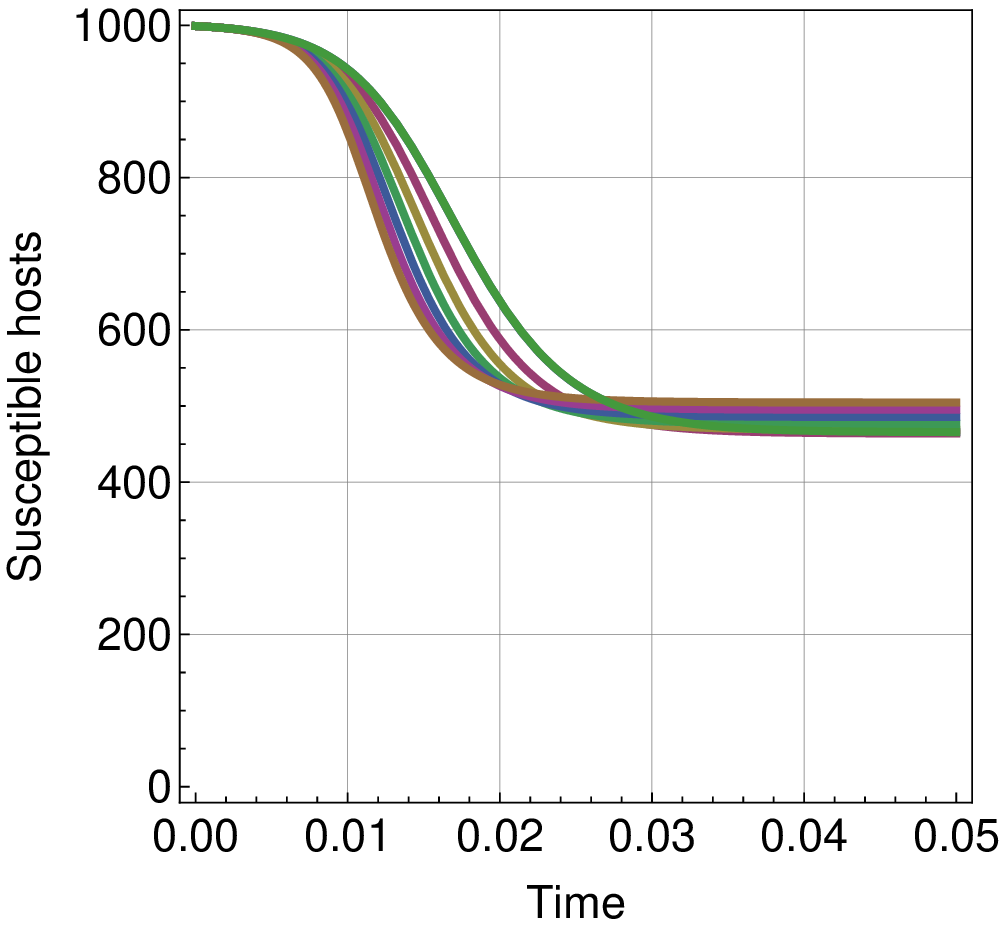}
\includegraphics[width=0.48\textwidth]{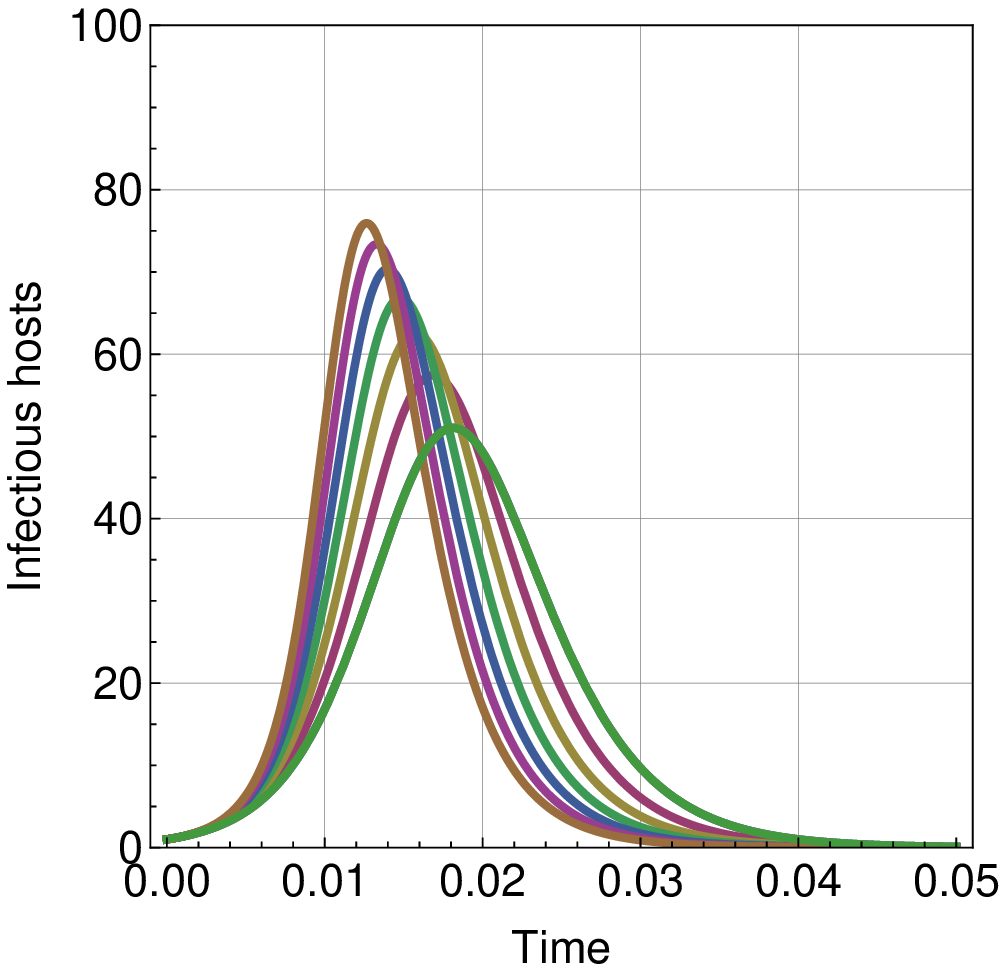}
\caption{Solutions to the heterogeneous SIR model \eqref{si1} with distributed infectivity and susceptibility . The parameters are $S_0=999,\,I_0=1,\,\gamma=700,\,\bar{\beta_1}(0)=1$ and the initial distribution of susceptibility is the gamma distribution with $\sigma^2_1(0)=0,0.2,0.4,0.6,0.8,1,1.2$, the initial distribution of infectivity is the gamma distribution with $\sigma^2_2(0)=0,0.005,0.01,0.015,0.02,0.025,0.03$}\label{fig:4}
\end{figure}

\subsection{Heterogeneous SIR model with distributed contact number}
Here we start with system \eqref{eq3a:4}:
\begin{equation}\label{eq4:14}
\begin{split}
\pdt s(t,w)&=-r\omega s(t,\omega)\mu(t),\\
\pdt i(t,w)&=-r\omega s(t,\omega)\mu(t)-\gamma i(t,\omega)\,,
\end{split}
\end{equation}
where
$$
\mu(t)=\frac{\int_{\Omega}\omega i(t,\omega)d\omega}{\int_{\Omega}\omega s(t,\omega)\,d\omega+\int_{\Omega}\omega i(t,\omega)\,d\omega}\,.
$$
First we note that system \eqref{eq4:14} is not covered by Theorem \ref{th2:1}, therefore, at least using the framework from Section \ref{section:2}, we cannot reduce this system to an ODE system. However, some results are still possible to obtain. In particular, we give a simple derivation of the expression for the initial growth of $\mu(t)$ (cf. \cite{May1988}).

If we denote $K(t)=\int_\Omega \omega i(t,\omega)$, then
$$
K'(t)=r\mathsf{E}[\omega^2]S\mu(t)-\gamma K(t).
$$
Differentiating $\mu(t)$ yields
\begin{equation}\label{eq4:15}
\begin{split}
  \mu'(t)&=\frac{K'(t)}{\int_{\Omega}\omega s(t,\omega)\,d\omega+\int_{\Omega}\omega i(t,\omega)\,d\omega}+\gamma \mu^2(t)=\\
  &=\frac{r\mathsf{E}[\omega^2]S\mu(t)}{\int_{\Omega}\omega s(t,\omega)\,d\omega+\int_{\Omega}\omega i(t,\omega)\,d\omega}-\gamma \mu(t)+\gamma \mu^2(t)=\\
  &=\mu(t)\left(\frac{r \mathsf{E}[\omega^2]S}{\mathsf{E}[\omega]S+K(t)}-\gamma\right)+\gamma \mu^2(t)=\\
  &=\mu(t)\left(\frac{r \mathsf{E}[\omega^2]}{\mathsf{E}[\omega]}-\gamma\right)+\gamma \mu^2(t),
\end{split}
\end{equation}
where the last equality holds for $t\to 0$ because $S(t)\to N$ and $K(t)\to 0$. In words, we obtained the well known results \cite{May1988} that initially the change in the number of susceptible hosts is proportional not to the mean number of contacts but to the mean number plus the coefficient of variation. This shows that the individuals who have a high number of contacts contribute disproportionally to the spread of an epidemic.

Corollary \ref{corr:14} implies that in case of distributed susceptibility or infectivity the initial phase of an epidemic can be determined by the first two moments. Here we prove that this is not so for model \eqref{eq4:14}, as it was shown in \cite{Veliov2005}. The proof uses Theorem \ref{th2:1}, direct proof is given in \cite{Veliov2005}. We use the notation $L(t)=\int_\Omega \omega s(t,\omega)\,d\omega$.

\begin{lemma}Let $S_1(t), \, S_2(t)$ be the solutions to model \eqref{eq4:14} with the initial conditions such that $\sigma_1^2(t)>\sigma_2^2(t)$ for the distribution of the contact number, all other initial conditions being equal. If $L(0)<K(0)$ then there exists $\varepsilon>0$ such that $S_1(t)>S_2(t)$ for all $t\in(0,\varepsilon)$. If $L(0)>K(0)$ then the opposite holds.
\end{lemma}
\begin{proof}From the first equation in \eqref{eq4:14} it follows
\begin{equation}\label{eq4:16}
\begin{split}
S''&=-r \mathsf{E}'[\omega]S\mu(t)-\mu(t)\mathsf{E}[\omega]S'-\mathsf{E}[\omega]S\mu'(t)=\mbox{(from \eqref{eq2:12} and \eqref{eq4:15})}\\
&=r\sigma^2(t) \mu^2(t)S-\mu(t)\mathsf{E}[\omega]S\left[\frac{r \mathsf{E}[\omega^2]S}{L(t)+K(t)}-\gamma\right]+\ldots=\\
&=\sigma^2(t)\mu^2(t)S\left[1-\frac{L(t)}{K(t)}\right]+\ldots,
\end{split}
\end{equation}
where dots denote terms that depend only on the first moment of the contact distribution.

From \eqref{eq4:16} lemma follows.
\end{proof}

As it was mentioned, Theorem \ref{th:2} cannot be applied to \eqref{eq4:14}. However, if we consider the case $\gamma=0$, we still can reduce the heterogeneous model to ODE. We rewrite equation \eqref{eq3a:5} in the form
\begin{equation}\label{eq5:1}
\pdt s(t,w)=-r\omega
s(t,\omega)\left[1-\frac{\bar{\omega}(t)S(t)}{C}\right],
\end{equation}
where $C$ is the number of contacts, which are made by the total
population, $\bar{\omega}(t)$ is the average number of contacts made
by one susceptible individual at time $t$. Theorem \ref{th:2} implies that \eqref{eq5:1} is equivalent to the following
ODE:
\begin{equation}\label{eq5:2}
\dot S(t)=-rh(S)\left[1-\frac{h(S)}{C}\right],
\end{equation}
where $h(S)$ is given by \eqref{eq4:2}. For example, if the initial distribution of the number of contacts of susceptible hosts is a gamma distribution with parameters $k$ and $\nu$, equation \eqref{eq5:2} takes the form
\begin{equation}\label{eq5:3}
\dot S(t)=-r\frac{k}{\nu}\left[\frac{S}{S_0}\right]^{1/k}S\left(1-\frac{k}{\nu}\left[\frac{S}{S_0}\right]^{1/k}\frac{S}{C}\right).
\end{equation}

Numerical solutions of \eqref{eq5:2} for two parameter sets are given in Figs. \ref{fig:5} and \ref{fig:6}. These figures show that initially the heterogeneous contact rates increase the speed of an epidemic; in the long term run, however, the presence of individuals who make a small number of contacts, slows the epidemic down.
\begin{figure}
\centering
\includegraphics[width=0.7\textwidth]{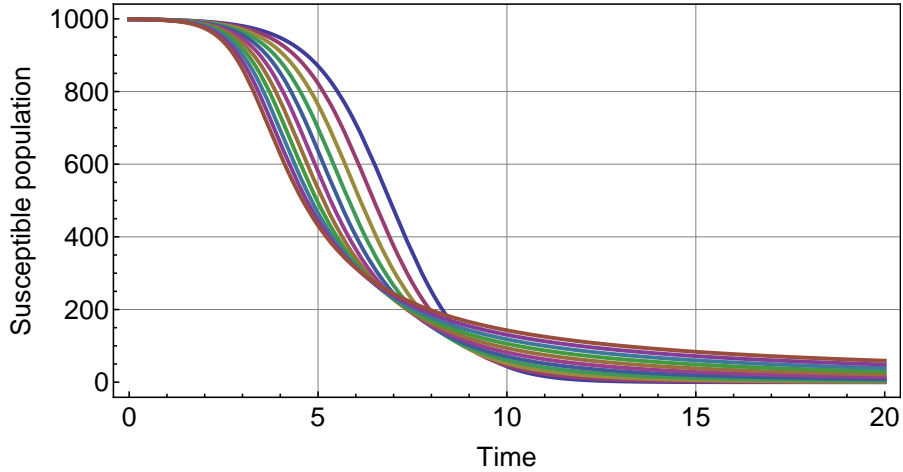}
\caption{Solutions to the heterogeneous SI model \eqref{eq5:3} with distributed contact number. The parameters are $S_0=999,\,I_0=1,\bar{\omega}(0)=1$ and the initial distribution of contact number is the gamma distribution with $\sigma^2(0)=0,0.1,0.2,0.3,0.4,0.5,0.6,0.7,0.8,0.9,1$}\label{fig:5}
\end{figure}
\begin{figure}
\centering
\includegraphics[width=0.7\textwidth]{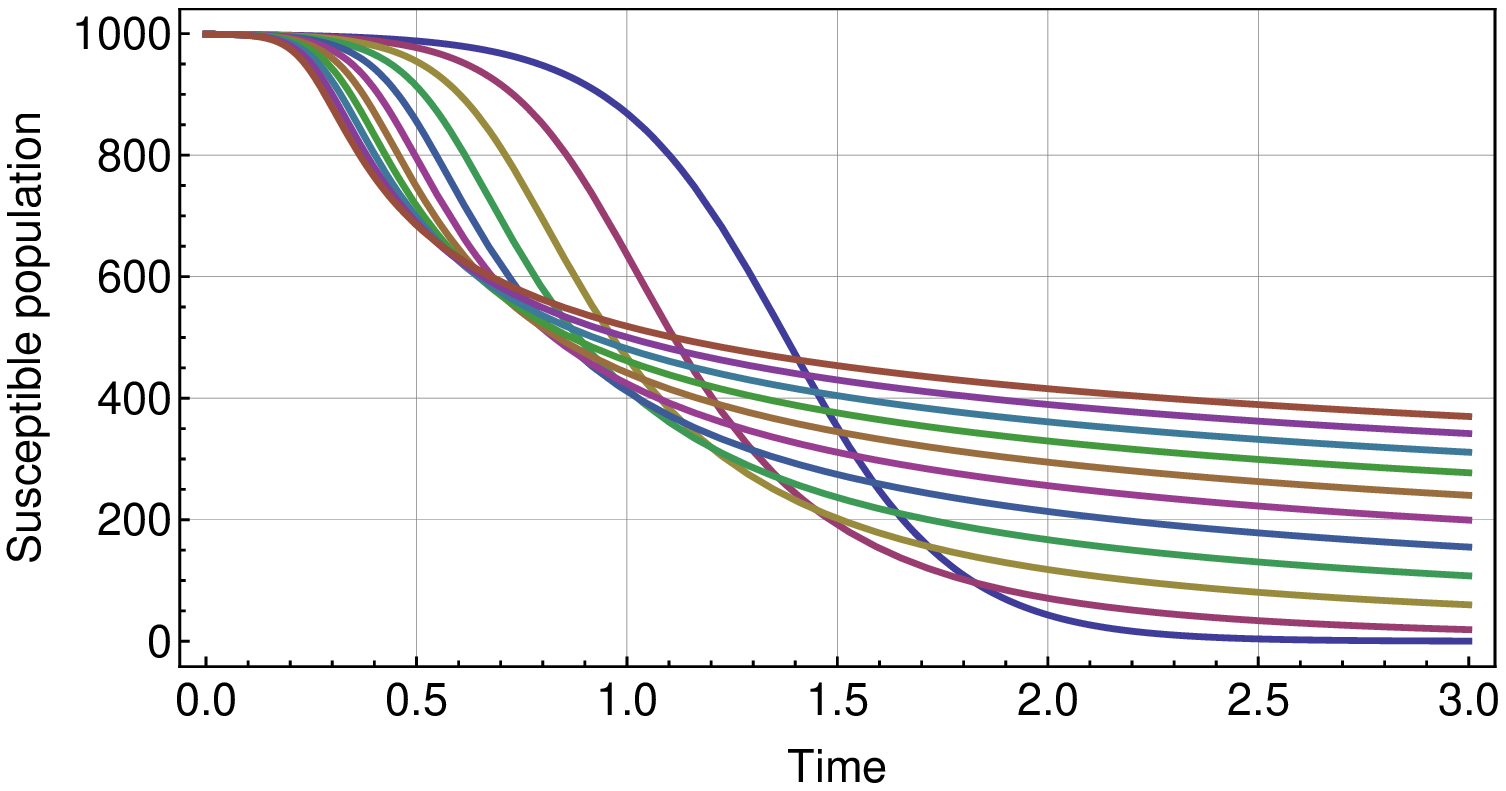}
\caption{Solutions to the heterogeneous SI model \eqref{eq5:3} with distributed contact number. The parameters are $S_0=999,\,I_0=1,\bar{\omega}(0)=5$ and the initial distribution of contact number is the gamma distribution with $\sigma^2(0)=0,10,20,30,40,50,60,70,80,90,100$}\label{fig:6}
\end{figure}

\section{Concluding comments}\label{conclusion} In this manuscript we reviewed and generalized several known results on the spread of epidemics in a closed heterogeneous population obtained with the help of the theory of heterogeneous populations with parametric heterogeneity. Our discussion was mainly about the SIR model \eqref{si1} and its particular cases \eqref{eq3:3}, \eqref{eq3:8}, and \eqref{si2}. We also introduced the methods of the theory of heterogeneous populations to the models that take into account non-uniform contact structure of the population, noting that explicit results can be inferred only for a very limited case of the SI model \eqref{eq5:1}.

While speaking of heterogeneity of the populations experiencing a disease, three different sources of heterogeneity can be accounted for. First, this is the heterogeneity in disease parameters that we termed the parametric heterogeneity, and we discussed this at length in the main text. Secondly, this is the heterogeneity of the social structure, which results in a different contact rates for different individuals. Finally, there is a third source of heterogeneity: the distribution of the infection length. For most models in the text it was tacitly assumed that this distribution is exponential with the mean $1/\gamma$ for all the infectious individuals. This assumption is usually made to simplify the mathematics and does not correspond to real situations. Using the approach from Section \ref{section:2} it is possible to slightly generalize our models by assuming that parameter $\gamma$ is also distributed through the population. Proceeding along the lines of Section \ref{section:3}, we obtain
\begin{equation}\label{eq6:1}
    \begin{split}
      \dot{S} &= -\beta SI, \\
       \pdt i(t,\omega) & =\beta S i(t,\omega)-\gamma(\omega)i(t,\omega),
    \end{split}
\end{equation}
which is equivalent, according to Theorem \ref{th2:1}, to the system
\begin{equation}\label{eq6:2}
    \begin{split}
      \dot{S} &= -\beta SI, \\
      \dot{I}& =\beta S I-\partial_\lambda \ln \mathsf{M}(\lambda)|_{\lambda=-t}I,
    \end{split}
\end{equation}
 with the mgf $\mathsf{M}(\lambda)$ of the initial distribution of $\gamma$. Systems \eqref{eq6:1} and \eqref{eq6:2} assume that the population consists of subpopulations, each of which has an exponentially distributed infection length, but this length varies from group to group according to the given initial distribution. A somewhat more interesting approach is to consider instead of \eqref{eq6:1} the following system
\begin{equation}\label{eq6:3}
    \begin{split}
      \dot{S} &= -\bar{\beta}(t) SI, \\
       \pdt i(t,\omega) & =\beta(\omega) S i(t,\omega)-\gamma(\omega)i(t,\omega),
    \end{split}
\end{equation}
where $\beta(\omega)$ and $\gamma(\omega)$ are correlated. Model \eqref{eq6:3} is not covered by Theorem \ref{th2:1}, but still can be tackled with the general theory from \cite{karev2009}.

In the mathematical models considered in the text it was always assumed that the population size is closed. This assumption allowed to formulate the models in the form suitable for Theorem \ref{th2:1}. It is a natural extension to consider models, when the demography and immigration processes are taken into account. This is very important because heterogeneous susceptibility of many diseases can be explained by a heritable genetic basis.  For example, we can consider the simplest SIR model with distributed susceptibility and recruitment in the form
\begin{equation}\label{eq6:4}
    \begin{split}
      \pdt s(t,\omega) &= \Lambda s(t,\omega)-\beta(\omega)s(t,\omega)I, \\
       \dot{I} & =\bar{\beta}(t)SI-\gamma I.
    \end{split}
\end{equation}
Model \eqref{eq6:4} can be reduced to an equivalent ODE system using Theorem \ref{th2:1}. However, the long-term behavior of \eqref{eq6:4} is straightforward: the individuals with higher values of the susceptibility parameter will outcompete those with smaller ones. For the model to be realistic it is also required to include the stochastic hereditary element --- mutations. That is, a more realistic counterpart of \eqref{eq6:4} is probably
\begin{equation}\label{eq6:5}
    \begin{split}
      \pdt s(t,\omega) &= \Lambda \int_\Omega \varphi(\omega,\eta) s(t,\eta)\,d\eta-\beta(\omega)s(t,\omega)I, \\
       \dot{I} & =\bar{\beta}(t)SI-\gamma I,
    \end{split}
\end{equation}
where $\varphi(\omega,\eta)$ gives the probability that a parent with the parameter value $\eta$ produces offspring with parameter value $\omega$. Model \eqref{eq6:5}, however, is not covered by Theorem \ref{th2:1}. We also note that a very similar to \eqref{eq6:4} model was considered in \cite{duffy2007rapid}, where, to solve the system, it was conjectured that the equation for the mean parameter value is (cf. \eqref{eq2:12})
$$
\dot{\bar{\beta}}(t)=-I\sigma^2(0).
$$
This means that is was implicitly assumed that the susceptibility distribution is normal, because the normal distribution is the one that satisfies the condition $\sigma^{2}(t)=\sigma^2(0)$.

Finally, all the models we discussed so far are deterministic. There is long history of studying stochastic SIR epidemics with multiple classes of susceptible and infectious individuals, e.g., \cite{andersson2000stochastic,ball1985deterministic,ball1993final,scalia1986asymptotic}, which is usually centered around the asymptotic distributions of the final epidemic size. We announce here that some of the methods presented in this text can be used for studying stochastic models \cite{novozhilov2012}.

\paragraph{Acknowledgements} The research is supported in part by the Russian Foundation for Basic Re-
search grant \# 10-01-00374. The author supported by the grant to young researches from Moscow State
University of Railway Engineering.

%\bibliography{infections}

%\begin{thebibliography}{10}
%\end{thebibliography}

\end{document}